\documentclass[submission,copyright,creativecommons]{eptcs}
\providecommand{\event}{Sixteenth International Symposium on \\
Games, Automata, Logics, and Formal \\ Verification (GandALF 2025)} 

\usepackage{iftex}
\usepackage{amssymb}
\usepackage[todo=show]{generic}
\usepackage{generic}
\usepackage{cleveref}

\usepackage[noamsdef,chgbar,linnum]{fmocdmac} 

\newcommand\gabriele[2][]{{\color{black}\todo[color=cyan!30,#1]{{\bf GP:} #2}}\ignorespaces}
\newcommand\angelo[2][]{{\color{black}\todo[color=red!30,#1]{{\bf AM:} #2}}\ignorespaces}

\usrmth{D}{}{sym}[D\,]
\usrmth{W}{}{sym}[W\,]
\usrmth{C}{}{sym}[C\,]
\usrmth{B}{}{sym}[B\,]

\newtheorem*{conjecture}{Conjecture}

\RenewDocumentCommand{\diamond}{}{\Diamond}
\RenewDocumentCommand{\square}{}{\Box}
\RenewDocumentCommand{\bigcirc}{}{\heartsuit}

\NewDocumentCommand{\logic}{m m m m}{
	\ensuremath{\mathrm{#1}%
	            \IfBlankF{#2}{^{#2}}%
	            \IfBlankF{#3}{_{#3}}%
	            \IfBlankF{#4}{#4}}\xspace}

\NewDocumentCommand{\FO}{O{}}{\logic{FO}{}{}{\IfBlankF{#1}{[#1]}}}

\NewDocumentCommand{\CTL}{s D(){} O{}}{\IfBooleanTF{#1}{%
	\logic{CTL}{\star\IfBlankF{#2}{,#2}}{#3}{}}{%
	\logic{CTL}{#2}{#3}{}}}

\NewDocumentCommand{\cCTL}{s D(){} O{}}{\IfBooleanTF{#1}{%
	\logic{cCTL}{\star\IfBlankF{#2}{,#2}}{#3}{}}{%
	\logic{cCTL}{#2}{#3}{}}}

\NewDocumentCommand{\ECTL}{s D(){} O{}}{\IfBooleanTF{#1}{%
	\logic{ECTL}{\star\IfBlankF{#2}{,#2}}{#3}{}}{%
	\logic{ECTL}{#2}{#3}{}}}

\NewDocumentCommand{\cCTLp}{s D(){} O{}}{\IfBooleanTF{#1}{
	\ensuremath{\cCTL*(p\IfBlankF{#2}{,#2})[#3]}\xspace}{%
	\ensuremath{\cCTL(p\IfBlankF{#2}{,#2})[#3]}\xspace}}

\NewDocumentCommand{\polcCTL}{s D(){} O{}}{\IfBooleanTF{#1}{%
	\ensuremath{\cCTL*(#2)[\pm\IfBlankF{#3}{,#3}]}\xspace}{%
	\ensuremath{\cCTL(#2)[\pm\IfBlankF{#3}{,#3}]}\xspace}}

\NewDocumentCommand{\polcCTLp}{s D(){} O{}}{\IfBooleanTF{#1}{
	\ensuremath{\polcCTL*(p\IfBlankF{#2}{,#2})[#3]}\xspace}{%
	\ensuremath{\polcCTL(p\IfBlankF{#2}{,#2})[#3]}\xspace}}

\NewDocumentCommand{\LTL}{D(){} O{}}{\logic{LTL}{#1}{#2}{}}

\NewDocumentCommand{\ELTL}{D(){} O{}}{\logic{ELTL}{#1}{#2}{}}

\NewDocumentCommand{\SafeLTL}{D(){} O{}}{\logic{SafeLTL}{#1}{#2}{}}

\NewDocumentCommand{\CosafeLTL}{D(){} O{}}{\logic{coSafeLTL}{#1}{#2}{}}

\cmdtxtoparname{PDL}

\DeclareRobustCommand{\ECTL}
  {{\logic{E}{}{}{}}\CTL}

\cmdtxtoparname{ATL}

\cmdtxtoparname{ATLP}[ATL$^{+}$]

\cmdtxtoparname{ATLS}[ATL*]

\cmdtxtoparname{PL}

\DeclareRobustCommandx{\OGPL}[3][1=, 2=, 3=]
  {\PL[#1][#2][1g\arglef{,}{#3}]}

\DeclareRobustCommandx{\CGPL}[3][1=, 2=, 3=]
  {\PL[#1][#2][cg\arglef{,}{#3}]}

\DeclareRobustCommandx{\DGPL}[3][1=, 2=, 3=]
  {\PL[#1][#2][dg\arglef{,}{#3}]}

\DeclareRobustCommandx{\AGPL}[3][1=, 2=, 3=]
  {\PL[#1][#2][ag\arglef{,}{#3}]}

\DeclareRobustCommandx{\EGPL}[3][1=, 2=, 3=]
  {\PL[#1][#2][eg\arglef{,}{#3}]}

\DeclareRobustCommandx{\BGPL}[3][1=, 2=, 3=]
  {\PL[#1][#2][bg\arglef{,}{#3}]}

\DeclareRobustCommandx{\XGPL}[3][1=, 2=, 3=]
  {\PL[#1][#2][xg\arglef{,}{#3}]}

\cmdtxtoparname{SL}

\DeclareRobustCommandx{\OGSL}[3][1=, 2=, 3=]
  {\SL[#1][#2][1g\arglef{,}{#3}]}

\DeclareRobustCommandx{\CGSL}[3][1=, 2=, 3=]
  {\SL[#1][#2][cg\arglef{,}{#3}]}

\DeclareRobustCommandx{\DGSL}[3][1=, 2=, 3=]
  {\SL[#1][#2][dg\arglef{,}{#3}]}

\DeclareRobustCommandx{\AGSL}[3][1=, 2=, 3=]
  {\SL[#1][#2][ag\arglef{,}{#3}]}

\DeclareRobustCommandx{\EGSL}[3][1=, 2=, 3=]
  {\SL[#1][#2][eg\arglef{,}{#3}]}

\DeclareRobustCommandx{\BGSL}[3][1=, 2=, 3=]
  {\SL[#1][#2][bg\arglef{,}{#3}]}

\DeclareRobustCommandx{\NGSL}[3][1=, 2=, 3=]
  {\SL[#1][#2][ng\arglef{,}{#3}]}

\DeclareRobustCommandx{\XGSL}[3][1=, 2=, 3=]
  {\SL[#1][#2][xg\arglef{,}{#3}]}

\ifpdf
  \usepackage{underscore}         
\else
  \usepackage{breakurl}           
\fi

\title{An Automaton-based Characterisation of \\ First-Order Logic over Infinite Trees}
\author{%
  Massimo Benerecetti%
  \institute{Università degli Studi di Napoli Federico II}%
  \email{massimo.benerecetti@unina.it}%
\and
  Dario Della Monica%
  \institute{Università degli Studi di Udine}%
  \email{dario.dellamonica@uniud.it}%
\and
  Angelo Matteo%
  \institute{Università degli Studi di Udine}%
  \email{matteo.angelo@spes.uniud.it}%
\and
  Fabio Mogavero%
  \institute{Università degli Studi di Napoli Federico II}%
  \email{fabio.mogavero@unina.it}%
\and
    Gabriele Puppis%
  \institute{Università degli Studi di Udine}%
  \email{gabriele.puppis@uniud.it}%
}
\def\titlerunning{An Automaton-based Characterisation of First-Order Logic over Infinite Trees}
\def\authorrunning{Benerecetti, Della Monica, Matteo, Mogavero and Puppis}

\begin{document}
\nolinenumbers
\maketitle



\begin{abstract}
In this paper, we study First Order Logic (\FO) 
over (unordered) infinite trees and its
connection with branching-time temporal logics.
More specifically, we provide an automata-theoretic characterisation of \FO
interpreted over infinite trees.
To this end, two different classes of hesitant tree automata are introduced and
proved to capture precisely the expressive power of two branching time temporal 
logics, denoted \polcCTLp and \cCTL*[f], which are, respectively, a restricted version of
counting \CTL with past and counting \CTL* over finite paths, both of which
have been previously shown equivalent to \FO over infinite trees.
The two automata characterisations naturally lead to normal forms for the two
temporal logics, and highlight the fact that \FO can only express properties of the
tree branches which are either safety or co-safety in nature.
\end{abstract}



\section{Introduction}

\emph{Characterisation theorems}~\cite{Flu85} are powerful model-theoretic tools
that offer a principled approach to understanding 
the intrinsic features of formal systems.
They allow us to mark the \emph{expressive boundaries} of specification
languages, compare these formalisms \wrt their \emph{descriptive power} on
specific classes of models, and design new languages starting from a given set
of requirements, in the spirit of \emph{Lindstr\"om-style theorems}~\cite{Lin69}
(\eg, based on maximality principles).
They also play a central role in \emph{definability theory}, guiding the
identification of expressive fragments and meaningful extensions of known
logics, thus supporting the selection of suitable languages for the
specification of the correct behaviour of systems in verification and synthesis
tasks.
\\\indent
A foundational distinction exists between \emph{linear-time} and
\emph{branching-time} languages~\cite{MP92,MP95}.
The former capture properties of computations viewed as totally-ordered sets of
events, while the latter account for the branching structure inherent in
concurrent and nondeterministic system behaviours.
\\\indent
The linear-time case, where models are isomorphic to (finite or infinite)
\emph{words}, is by now well understood.
A rich and intertwined network of equivalences connects \emph{predicate logics}
over linear orders with \emph{temporal logics}, such as \LTL~\cite{Pnu77,Pnu81}
and \ELTL~\cite{Wol83}, with \emph{star-free}~\cite{Sch65,PP86} and
\emph{$\omega$-regular}~\cite{Buc62} languages, and with automata-theoretic
models, including \emph{finite}~\cite{Ner58,RS59} and
\emph{B\"uchi}~\cite{Elg61,Buc62,Buc66} automata.
These connections provide deep insights into the structure of definable
properties and lead to optimal decision procedures across different
representations.
\\\indent
By contrast, the branching-time setting remains more fragmented.
Even for \emph{First-Order Logic} (\FO) interpreted over (finite or infinite)
trees many fundamental definability questions remain unsettled.
A longstanding open problem posed by Thomas in the 1980s~\cite{Tho84} asks
whether it is decidable if a given regular-tree language is definable in \FO.
This question has been studied under various combinations of tree types
(\emph{ranked/unranked}, \emph{ordered/unordered}) and interpreted vocabularies
(\eg, including only \emph{child}, only \emph{ancestor}, or both relations).
Aside from the positive result for \FO over finite trees with the child
relation~\cite{BS09}, the problem remains open in all other settings.
Efforts to resolve this question have mainly followed \emph{algebraic
approaches}~\cite{Boj21}, inspired by their success in the word case (most
notably the Sch\"utzenberger theorem on star-free languages~\cite{Sch65}).
These approaches rely on the compositionality and structural insights provided
by syntactic algebras.
Despite significant progress, they have provided only partial results, mostly
for classes of finite trees~\cite{Pot95,Boj04,EW10,BSW12} or topologically simple infinite
trees~\cite{BI09, BIS13}.
An alternative and often complementary line of work seeks direct
characterisations of \FO-definable tree languages via automata.
This route, highly successful in the linear-time case, has also led to fruitful
results in the branching-time setting, including a correspondence~\cite{JW96}
between \emph{Monadic Second-Order Logic} (\MSO)~\cite{Rab69}, \emph{Parity Tree
Automata}~\cite{Mos84,EJ91}, and the \emph{Modal \MC}~\cite{Koz83}.
More recently~\cite{BBMP23}, the landscape has expanded to include the
expressive equivalence of \emph{Monadic Chain/Path Logics}
(\MCL/\MPL)~\cite{Tho84,Tho87,HT87}, their temporal \emph{Computation Tree
Logic} counterparts (\ECTL*/\CTL*)~\cite{VW83,EH83,EH85}, and variants of
\emph{Hesitant Tree Automata} (HTA)~\cite{KVW00}.
\\\indent
In this work, we continue 
this line of development, by providing the
first, to the best of our knowledge, complete automata-theoretic
characterisation of first-order logic with the descendant 
relation of unranked unordered infinite trees.
Our approach builds on previous results for two branching-time temporal logics,
namely a \emph{fragment of Computation Tree Logic with past and counting}, 
denoted \emph{\polcCTLp},
and the \emph{Full Computation Tree Logic with counting and finite path quantification}, 
denoted \emph{\cCTL*[f]}. In~\cite{Sch92a,BBMP23} these logics were shown to be expressively 
equivalent to \FO when interpreted on unordered infinite trees.
For 
these two logics, we introduce 
corresponding variants of hesitant graded automata, 
called \emph{Two-Way Hesitant Linear Tree Automata} (2HLGT) and
\emph{counter-free Hesitant Weak Tree Automata} (HWGT$_{cf}$), 
and prove that they capture precisely the expressive power of the
considered logics, and therefore of \FO as well. This establishes
a full mutual equivalence between logics and automata.
These characterisations also uncover a \emph{polarised normal form} for both
temporal logics, revealing a noteworthy semantic feature of \FO over infinite
trees: formulas that quantify existentially on branches can only express
\emph{co-safety} properties, while those quantifying universally correspond
exclusively to \emph{safety} properties.
This observation aligns with earlier findings~\cite{CV14} that relate 
fragments of the modal \MC, variants of \emph{Propositional Dynamic Logic}
(\PDL)~\cite{FL79}, and \emph{Weak Monadic Chain Logic} (\WMCL).

\noindent{\bf Other related work.}
In earlier work, Boja{\'n}czyk~\cite{Boj04} showed that, over finite binary
trees, \FO with child and ancestor relations is equivalent to a
\emph{cascade product} of so-called \emph{aperiodic wordsum automata}.
While related in spirit, this result targets a different logic and a different
class of structures.
More recently, Ford~\cite{For19} focused on the same tree structures that
are considered here, and introduced the class of 
\emph{antisymmetric path parity automata}, which are shown to be no more 
expressive than \FO.
However, that work does not provide a translation from \FO to automata, 
leaving the equivalence question open.

\noindent{\bf Organization.}
The paper is organised as follows. In Section 2 we give the necessary
preliminaries on words, trees, first order and temporal logics. In Section 3 we
recall the two branching-time temporal logics equivalent to \FO and investigate
their most interesting properties. Section 4 is devoted to the introduction of
the class of \emph{graded tree automata} and of its weak and hesitant
restrictions. Sections 5 and 6 are the main sections, in which we prove the
equivalence of two classes of automata with the two temporal logics discussed in
Section 3, while in Section 7 we present the normal forms obtained for them.
Finally, Section 8 discusses the results.

\section{Preliminaries}

\paragraph{Words and trees.}
Given a finite alphabet $\Sigma$, a \emph{finite} (resp., \emph{infinite}) word over $\Sigma$ is a finite (resp. infinite) sequence 
%
%
of letters in $\Sigma$.
%
%
A \emph{word language} over $\Sigma$ is a set of words over
$\Sigma$.

We consider unranked and unordered infinite trees with arbitrary finite
branching. A \emph{tree} is a connected acyclic graph $T = (V_T, E_T)$, where
$V_T$ is its set of \emph{nodes} and $E_T \subseteq V_T \times V_T$ is its
\emph{transition} relation; $E_T$ is total, that is, for every $n \in V_T$,
$(n,n') \in E_T$ for some $n'$. 
We denote the reflexive and transitive closure of $E_T$ by $E_T^*$.
We denote the root of a tree $T$ by $\epsilon_T$. Given two nodes $v,w$ in $T$,
we say that $v$ is the \emph{parent} of $w$ if $(v,w) \in E_T$ and that $v$ is a
\emph{child} or \emph{successor} of $w$ if $(w,v) \in E_T$. A tree is \emph{unranked} if there is no restriction (apart from finiteness) on the number of children a node might have. If two nodes have the same parent we say they are \emph{siblings}. A tree is \emph{unordered} if the order of the siblings is irrelevant. Moreover, we say
that $v$ is an \emph{ancestor} of $w$ if ($v,w) \in E_T^*$, and that $v$ is a
\emph{descendant} of $w$ if ($w,v) \in E_T^*$. 

The \emph{subtree} of $T$ rooted at a node $w$ is the tree consisting of all the
descendants of $w$. A \emph{path} $\pi$ of a tree $T$ is a finite or infinite
sequence of nodes of $T$, whose first element is the root of $T$ and where every
element but the first one is a child of its predecessor in the sequence.
Given a path $\pi = n_0 n_1 \ldots$, we write $\pi(i)$ to refer to $n_i$.

Given a finite alphabet $\Sigma$, a $\Sigma$-labelled tree ($\Sigma$-tree for
short) is a pair $\mathcal{T} = (T, \tau$) such that $T$ is a tree and $\tau :
V_T \rightarrow \Sigma$ is a labelling function assigning to each node in $T$ a
letter of $\Sigma$.
A \emph{tree language} over $\Sigma$ is a set of $\Sigma$-trees. We
denote by $\mathcal{T}_\Sigma$ the language of all $\Sigma$-trees.


\medskip

\noindent{\bf First-Order logic. }
We introduce \emph{monadic First-Order logic on trees with the ancestor relation} (\FO for short). 
Let $\mathrm{AP}$ be a set of atomic propositions. 
Formulas of \FO are generated by the following grammar, 
where $p$ ranges over $\mathrm{AP}$ and $x, x_0,x_1,\dots$
are first-order variables from a set $\mathrm{Var}$:
\begin{align*}
\varphi ::=\;&
  (x_0 = x_1)
  \;\bigm|\;
  (x_0 \le x_1)
  \;\bigm|\;
  p(x)
  \;\bigm|\;
  \neg \varphi
  \;\bigm|\;
  \varphi \lor \varphi
  \;\bigm|\;
  \exists x. \varphi
 \end{align*}
The usual abbreviations $\top, \bot, \land, \rightarrow, \leftrightarrow, \forall$
are allowed.
A variable is \emph{free} if it is  not bound to any quantifier.
A formula without free variables is called a \emph{sentence}.

An \FO formula $\varphi$ is interpreted over a structure
$\mathcal{M} = (\mathcal{T},\zeta$), where $\mathcal{T} = (T, \tau)$ is a
$\Sigma$-tree,%
\footnote{It is sometimes convenient to identify $\Sigma$ with $\mathrm{AP}$.
  Since here we are not concerned with succinctness of formulas or complexity of
  satisfiability problems, this detail is immaterial.}
with $\Sigma=2^{\mathrm{AP}}$, and $\zeta: \mathrm{Var} \rightarrow V_T$
is a function associating a node of $T$ with each variable.
The relational symbols $=$ and $\le$ are interpreted, respectively, by the
identity on the nodes and by the reflexive and transitive closure $E_T^*$ of the
transition relation of $T$; the rest of the semantics is standard.
Note that we can equally interpret \FO formulas on words, seen as trees whose
nodes have at most one successor.

%

\medskip

\noindent{\bf Temporal logics.}
We now introduce variants of the branching-time temporal logics \CTL and \CTL*.
In doing so, we adopt a suggestive notation that annotates a base logic, e.g.~$\CTL$,
with superscripts and subscripts, denoting enhancements and restrictions, respectively.
The most expressive logical formalism in this setting is denoted \cCTL*(p), and can be 
seen as the extension of \CTL* with counting modalities (e.g., for counting the number
of successors satisfying a certain property) and past operators 
(for navigating the tree along ancestors). 
Formulas of \cCTL*(p) are generated by the following grammar, 
where $n$ ranges over the natural numbers and $p$ ranges over 
a set $\mathrm{AP}$ of atomic propositions:
\begin{align*}
\varphi &::=
  \underbrace{{\D^n \varphi}}_{\scriptstyle\hspace{-8mm}\text{counting operators}\hspace{-8mm}}
  \;\bigm|\;
  p
  \;\bigm|\;
  \neg \varphi
  \;\bigm|\;
  \varphi \lor \varphi
  \;\bigm|\;
  {\E\varphi}
  \;\bigm|\;
  {\X \varphi}
  \;\bigm|\;
  {\varphi \U \varphi}
  \;\bigm|\;
  \underbrace{{\Y \varphi}
  \;\bigm|\;
  {\varphi \S \varphi}}_{\scriptstyle\text{past operators}}
\end{align*}
Some abbreviations are
$\F\varphi = \top \U \varphi$, 
$\G\varphi = \neg \F \neg \varphi$, 
$\varphi \R \varphi' = \neg(\neg\varphi \U \neg\varphi')$, 
$\varphi \W \varphi' = \G \varphi \lor (\varphi \U \varphi')$, 
$\A\varphi = \neg \E\neg\varphi$.

Formulas of \cCTL*(p) are evaluated with respect to a $\Sigma$-tree $\mathcal{T} = (T, \tau)$, where $\Sigma = 2^{\mathrm{AP}}$, 
an infinite path $\pi$ of $T$, and a position $i \in \mathbb N$ along this path.
The satisfaction relation $\models$ is defined by induction over \cCTL*(p) formulas
as follows (we omit the obvious Boolean cases):
\begin{compactitem}
    \item $\mathcal{T}, \pi, i \models \D^n\psi$ \iff 
          there are $n$ 
          distinct infinite paths $\pi_1,...,\pi_n$ such that
          \begin{inparaenum}[(i)]
          \item they coincide with $\pi$ up to position $i$,
          \item they all differ at position $i+1$, and
          \item they satisfy $\mathcal{T}, \pi_j, i+1 \models \varphi$
            for all $1 \leq j \leq n$,
          \end{inparaenum}
    \item $\mathcal{T}, \pi, i \models p$ \iff $p \in \tau(\pi(i))$,
    \item $\mathcal{T}, \pi, i \models \E\psi$ \iff there is an infinite path $\pi'$ that 
          coincides with $\pi$ up to position $i$ and satisfies $\mathcal{T}, \pi', i \models \psi$,
    \item $\mathcal{T}, \pi, i \models \X \psi$ \iff $\mathcal{T}, \pi, i +1 \models \psi$,
    \item $\mathcal{T}, \pi, i \models \psi \U \psi'$ \iff there is $j \geq i$ such that 
          $\mathcal{T}, \pi, j \models \psi'$ and $\mathcal{T}, \pi, k \models \psi$ for all $i \leq k < j$,
    \item $\mathcal{T}, \pi, i \models \Y \psi$ \iff $i > 0$ and $\mathcal{T}, \pi, i -1 \models \psi$,
    \item $\mathcal{T}, \pi, i \models \psi \S \psi'$ \iff there is $j \leq i$ such that 
          $\mathcal{T}, \pi, j \models \psi'$ and $\mathcal{T}, \pi, k \models \psi$ 
          for all $i \geq k > j$.
\end{compactitem}
We say that two formulas $\psi$ and $\psi'$ are \emph{equivalent}%
\footnote{This is a standard notion called \emph{initial equivalence}, sometimes
          contrasted with the stronger equivalence that evaluates formulas at arbitrary
          nodes. In this paper, we are only concerned with initial equivalence,
          and therefore we opt for saying just ``equivalent'' without further specification.} 
if for every $\Sigma$-tree $\mathcal{T} = (T, \tau)$ and every infinite path $\pi$ of $T$, we have that 
$\mathcal{T}, \pi, 0 \models \psi$ iff $\mathcal{T}, \pi, 0 \models \psi'$.
We say that a $\Sigma$-tree $\mathcal{T}$ is a \emph{model} of a formula $\varphi$, 
denoted $\mathcal{T} \models \psi$, if $\mathcal{T}, \pi, 0 \models \psi$ for every infinite path $\pi$. 
A formula is \emph{valid} if every $\Sigma$-tree is a model of it.  

We denote by $\mathcal{L}(\psi)$ the language of models of a given formula $\psi$, and we shall
consider classes of languages defined by formulas in certain fragments of \cCTL*(p). 
In particular, we say that two such fragments are \emph{expressively equivalent} if they define 
the same class of languages, and that one fragment is \emph{strictly less expressive} than 
another one when the class of languages defined by the former is strictly contained in the 
class of languages defined by the latter. 

Let us now discuss the main fragments of $\cCTL*(p)$.
A first fragment is \CTL*(p), which is \cCTL*(p) devoid of the counting
operators $\D^n$.
In its turn, $\CTL*(p)$ without the past operators corresponds to the classical
\CTL* logic. 
Fragments \CTL(p) and \CTL are obtained by applying the same restrictions
(analogous removals of counting and past operators) to $\cCTL(p)$, the fragment
of \cCTL*(p) where path quantifiers and future temporal operators must be paired
together, as indicated by the following grammar:
%
\begin{align*}
  &
  \phantom{
    ::=\;
    }
  ~\overbrace{\phantom{
    {\D^n\varphi
  }
  \;\bigm|\; }
  ~\overbrace{\phantom{
  p
  \;\bigm|\; 
  \neg \varphi
  \;\bigm|\; 
  \varphi \lor \varphi
  \;\bigm|\; 
  {\E\,\X\,\varphi}
  \;\bigm|\; 
  {\E\,( \varphi \,\U\, \varphi )}
  \;\bigm|\; 
  {\A\,( \varphi \,\U\, \varphi )}
  \;\bigm|\; 
  {\Y \varphi}
  \;\bigm|\; 
  {\varphi \,\S\, \varphi}
  }}^{\text{\CTL(p)}}}^{\cCTL(p)}
\\[-3.5ex]
  \varphi &::=\;
  \underbrace{
  {\D^n\varphi}
  \;\bigm|\;
  \underbrace{
  p
  \;\bigm|\; 
  \neg \varphi
  \;\bigm|\; 
  \varphi \lor \varphi
  \;\bigm|\; 
  {\E\,\X\,\varphi}
  \;\bigm|\; 
  {\E\,( \varphi \,\U\, \varphi )}
  \;\bigm|\; 
  {\E\,( \varphi \,\R\, \varphi )}
  }_{\text{\CTL}}
  }_{\text{\cCTL}}
  \;\bigm|\; 
  {\Y \varphi}
  \;\bigm|\; 
  {\varphi \,\S\, \varphi} 
\end{align*}

\noindent
Finally, \cCTL* and \cCTL are, respectively, \cCTL*(p) and \cCTL(p) without past
operators.
It is well known that \CTL is strictly less expressive than \CTL(p) \cite[Theorem 4.1, 4.2]{LS95},
and consequently \cCTL is strictly less expressive than \cCTL(p).

Finally, we have the Linear-time Temporal Logic with and without past operators,
abbreviated, respectively, \LTL(p) and \LTL for short. These logics can be seen
as the fragments of \CTL*(p) and \CTL* that are evaluated over unary trees, that
is words (notice that in this setting the path quantifier $\E$ becomes
pointless).
It is worth recalling that \LTL and \LTL(p) are expressively equivalent
\cite[Theorem 2.2]{GPSS80}, and they correspond to \FO when the latter is 
evaluated on words as well:

\begin{proposition}{\cite[Theorem 1.1]{DG08a}} \label{LTL and LTL+p are equivalent}
\LTL and \LTL(p) are equivalent to \FO over finite and infinite words.
\end{proposition}


\medskip

\noindent{\bf Word automata.}
A \emph{Nondeterministic Parity word Automaton} (\emph{NPA}) is a tuple 
$\mathcal{A} = \langle Q, \Sigma, \delta, q_I, \Omega\rangle$, 
where $Q$ is a finite set of states, $\Sigma$ is a finite alphabet, 
$\delta:  Q\times\Sigma\rightarrow 2^Q$
is a 
transition relation (represented as a function towards the powerset of $Q$), 
$q_I\in Q$ is an initial state, and
$\Omega: Q \rightarrow \mathbb{N}$ is a \emph{priority function} 
associating a number with each state. An NPA is \emph{deterministic} if $\delta: Q \times \Sigma \rightarrow Q$.

A \emph{path} of $\mathcal{A}$ on a finite (resp., infinite) word $w = a_0 a_1 \dots$
is a finite (resp., infinite) sequence of states $q_0 q_1 \dots$ such that
$q_{i+1}\in\delta(q_{i},a_{i})$ for all positions $i$ in $w$.
Such a path is called a \emph{run} of $\mathcal{A}$ on $w$ if it is infinite and, moreover, 
its first state $q_0$ is the initial state $q_I$ of the automaton.
We denote by $\inf(r)$ the set of states visited infinitely often along a run $r$, 
and say that $r$ is \emph{successful} if the highest priority of the states in $\inf(r)$ is even.
An infinite word $w$ is \emph{accepted} by $\mathcal{A}$ if there is a successful run
of $\mathcal{A}$ on $w$.
The language $\mathcal{L}(\mathcal{A})$ \emph{recognized} by $\mathcal{A}$ is
the set of words accepted by $\mathcal{A}$.
%
%
Two automata are \emph{equivalent} if they recognize the same language.

Variants of NPAs are obtained by changing the acceptance conditions.
One way is to simply turn existential non-determinism into universal non-determinism,
that is, to declare that the automaton accepts a given word whenever \emph{every} run 
on it is successful; we call these models of automata \emph{universal} (e.g., we have
\emph{Universal Parity word Automata}, abbreviated UPA).
Another way is to constrain the co-domain of the priority function to be of
cardinality 2 and to contain both an odd and an even number.
An NPA with this restriction is a \emph{B\"uchi} (resp., \emph{coB\"uchi})
automaton if the largest number in the co-domain $\mathit{cod}(\Omega)$ of
$\Omega$ is even (resp., odd).
Note that such an acceptance condition can be equivalently specified by
the set $F = \{ q \in Q \mid \Omega(q) = \max(\mathit{cod}(\Omega)) \}$, 
whose states are usually called \emph{accepting} or \emph{rejecting}
depending on whether we deal with B\"uchi or coB\"uchi conditions.
Accordingly, a run $r$ is successful for a B\"uchi (resp., coB\"uchi)
condition if $\inf(r) \cap F\neq \emptyset$ (resp., $\inf(r) \cap F = \emptyset$). 



An automaton $\mathcal{A}$ is \emph{counter-free} if the following holds 
for all states $q$, all finite words $\iota$, and all numbers $n$: 
if $\mathcal{A}$ admits a path on $\iota^n$ that starts and ends at 
$q$, then it admits a path also on $\iota$ that starts and ends at $q$. 
By \cite[Theorem 1.1]{DG08a}, counter-free B\"uchi nondeterministic automata and \FO 
capture the same class of word languages.



\section{Temporal Logics and FO}
In this section, we introduce two temporal logics provably equivalent to \FO
over unranked and unordered finitely branching infinite trees. 

\medskip

\noindent{\bf The polarized fragment of 
counting CTL with past operators.}
The first temporal logic we discuss is a syntactic restriction on
\cCTLp, introduced and studied by Schlingloff in \cite{Sch92a}.
More precisely, in \cite[Theorem 4.5]{Sch92a} it is shown that there is a fragment of 
\cCTL(p) that is expressively equivalent to \FO over infinite finitely branching trees.
We call this fragment \emph{Polarized \cCTL(p)}, abbreviated \polcCTL(p). Its formulas
are generated by the following grammar:
\begin{align*}
\varphi ::=\;&
  {\D[][n] \varphi}
  \;\bigm|\;
  p
  \;\bigm|\;
  \neg \varphi
  \;\bigm|\;
  \varphi \lor \varphi
  \;\bigm|\;
  {\E\X \varphi}
  \;\bigm|\;
  {\E(\varphi \U \varphi)}
  \;\bigm|\;
  {\Y \varphi}
  \;\bigm|\;
  {\varphi \S \varphi}
 \end{align*}
Accordingly, we denote by \polcCTL the sub-fragment obtained from \polcCTL(p) by disallowing past operators.

The only difference w.r.t.~\cCTLp and~\cCTL is that $\E (\varphi \R \varphi)$
is not included as primitive in the syntax (and it cannot be restored as an abbreviation). 
Indeed, it is possible to define the following abbreviations: $\A\X \varphi = \neg \E\X \neg \varphi$, $\E\F \varphi = \E(\top \U \varphi)$, $\A\G \varphi = \neg \E\F
\neg \varphi$, $\A (\varphi \R \varphi') = \neg \E (\neg \varphi \U \neg
\varphi')$; however, formulas of the form $\E\G \varphi$, $\E(\varphi \R \varphi')$,
$\A\F\varphi$, $\A(\varphi \U \varphi')$ are not derivable.
The semantic intuition behind this syntax is that formulas that existentially quantify over branches can only express co-safety properties; dually, universal quantifications can only be paired with safety properties. This will become transparent in the following of the paper, but note already that this is similar to what happens in
a fragment of the modal \MC isolated in \cite{CV14}, that
the authors call \emph{completely additive}, and that allows only the interplay of least fixpoint operators and existential modalities, and, dually, of greatest fixpoints and universal modalities.

%
%

Before proceeding, we address two natural issues: whether \polcCTLp is 
strictly less expressive than \cCTLp and whether \polcCTL is strictly
less expressive than \polcCTLp. It turns out that the answer is positive in both
cases, for the following two propositions.
\begin{proposition}{\cite[Theorem 3]{BBMP23}}] \label{CTL+p more expressive than polctlp}
    No \FO formula can express the \CTL formula $\A\F p$.
\end{proposition}
\noindent The result follows from the equivalence of \FO and \polcCTLp.

\begin{proposition}{\cite[Lemma A.2]{LS95}} \label{polCTL stricly less than polCTLp}
    No \polcCTL formula can express the \CTL* formula $\E((p \lor q\U
    r) \U s)$.
\end{proposition}
\noindent By investigating the proof of \cite[Lemma
  A.1]{LS95}, it turns out that there is indeed a formula
of \polcCTLp equivalent to $\E((p \lor q\U r) \U s)$ (it is quite long so we omit
it here). Note how such a formula predicates about a \emph{finite prefix} of the
selected path. Finally, to sum up: \polcCTL is strictly less expressive than
\polcCTLp, which is in turn less expressive than \cCTLp.

\medskip

\noindent{\bf Counting CTL$^*$ 
over finite paths.}
The second temporal logic equivalent to \FO that we investigate is
\cCTL*[f]. The equivalence of this logic with \FO can be established
easily by adapting the model-theoretic argument of \cite{MR03} and
it was first noticed in \cite[Proposition 3]{BBMP23}.

\cCTL*[f] is \cCTL* with path quantification ranging over
\emph{finite paths}. Given a $2^{AP}$-tree
$\mathcal{T} = (T, \tau)$, a finite non-empty path $\pi$ of $T$, and a
position $i$ on the path, the satisfaction relation $\models$ of a
\cCTL*[f] formula is defined as follows (once again we omit the obvious Boolean
cases):
\begin{compactitem}
    \item $\mathcal{T}, \pi, i \models \D^n\psi$ \iff path $\pi$ does not end
      at $\pi(i)$ and
          there are $n$
          distinct finite paths $\pi_1,...,\pi_n$ such that
          \begin{inparaenum}[(i)]
          \item they coincide with $\pi$ up to position $i$,
          \item they all differ at position $i+1$, and
          \item they satisfy $\mathcal{T}, \pi_j, i+1 \models \varphi$
            for all $1 \leq j \leq n$,
          \end{inparaenum}
    \item $\mathcal{T}, \pi, i \models p$ \iff $p \in \tau(\pi(i))$,
    \item $\mathcal{T}, \pi, i \models \E\psi$ \iff there is a finite path
      $\pi'$ coinciding with $\pi$ up to $i$, such that $\mathcal{T}, \pi', i
      \models \psi$,
    \item $\mathcal{T}, \pi, i \models \X\psi$ \iff path $\pi$ does not end
      at $\pi(i)$ and $\mathcal{T}, \pi, i+1 \models \psi$.
    \item $\mathcal{T}, \pi, i \models \psi \U \psi'$ \iff there is $j \geq i$
      such that $\mathcal{T}, \pi, j \models \psi'$ and $\mathcal{T}, \pi, k
      \models \psi$ for all $i \leq k < j$,
\end{compactitem}
While the semantics for atomic propositions and the temporal operator until $\U$
is unchanged compared to \cCTL*, the other clauses must be slightly adapted to
deal with finite path quantification.
The notions of formula equivalence and model must be adapted accordingly: two
formulas $\psi$ and $\psi'$ are \emph{equivalent} if for every $\Sigma$-tree
$\mathcal{T} = (T, \tau)$ and every finite non-empty path $\pi$ of $T$, we have
that $\mathcal{T}, \pi, 0 \models \psi$ iff $\mathcal{T}, \pi, 0 \models \psi'$;
a $\Sigma$-tree $\mathcal{T}$ is a \emph{model} of a formula $\varphi$, denoted
$\mathcal{T} \models \psi$, if $\mathcal{T}, \pi, 0 \models \psi$ for every
finite non-empty path $\pi$.
%

In \cCTL*[f], it makes sense to define an abbreviation for the dual of $\X$,
defined as $\tilde{\X}\varphi = \neg \X\neg\varphi$, thus corresponding to the
semantic clause:
\begin{compactitem}
\item $\mathcal{T}, \pi, i \models \tilde{\X}\psi$ \iff the path $\pi$ ends
  at $\pi(i)$ or $\mathcal{T}, \pi, i+1 \models \psi$.
\end{compactitem}
Notice that, in \cCTL*, the semantics of $\X$ and $\tilde{X}$ coincide ($\X$ is
the dual of itself).
The introduction of $\tilde{\X}$
changes a few things. Recall that
$\F\varphi \leftrightarrow \varphi \lor \X\F\varphi$ and $\varphi \U \psi
\leftrightarrow \psi \lor (\varphi \land \X(\varphi \U \psi))$. Now, by putting
$\G\varphi = \neg \F\neg\varphi$ and $\varphi \R \psi = \neg(\neg\varphi \U \neg
\psi)$, one obtains $\G\varphi \leftrightarrow \varphi \land
\tilde{\X}\G\varphi$ and $\varphi \R \psi \leftrightarrow \psi \land (\varphi
\lor \tilde{\X}(\varphi \R \psi))$.
Moreover, $\neg \E\X \neg \psi =
\A\tilde{\X}\psi$.

The equivalence of \cCTL*[f] with \polcCTL(p) could come as a surprise, 
since the former does not share the syntactic restriction of the latter. In the following proposition we show how finite path quantification
actually enforces semantically the same behaviour of \polcCTL(p).

\begin{proposition} \label{ctlf equivalences}
    Let $\varphi$, $\psi$, and $\gamma$ be \cCTL*[f] formulas. Then,
    the following are valid \cCTL*[f] equivalences:
    \begin{compactenum}
    \item $\E \tilde{\X} \varphi \leftrightarrow \top$ and
      $\A \X \varphi \leftrightarrow \bot$,
    \item $\varphi \R \psi \leftrightarrow \alpha^{\R}_{\varphi,\psi}$ and
      $\varphi \U \psi \leftrightarrow \alpha^{\U}_{\varphi,\psi}$, where
      $\alpha^{\R}_{\varphi,\psi} = \psi \U ((\tilde{X} \bot \vee \varphi)
      \wedge \psi)$ and $\alpha^{\U}_{\varphi,\psi} = \psi \R ((\X \top \wedge
      \varphi) \vee \psi)$
    \item $\E(\varphi \R \psi) \leftrightarrow \E(\alpha^{\R}_{\varphi,\psi})$
      and $\A(\varphi \U \psi) \leftrightarrow \A(\alpha^{\U}_{\varphi,\psi})$,
%
    \item $\E\X(\varphi \R \psi) \leftrightarrow
      \E\X(\alpha^{\R}_{\varphi,\psi})$ and $\A\tilde{\X}(\varphi \U \psi)
      \leftrightarrow \A\tilde{\X}(\alpha^{\U}_{\varphi,\psi})$
    \item $\E (\varphi \U (\psi \R \gamma)) \leftrightarrow \E(\varphi \U
      \alpha^{\R}_{\psi,\gamma})$ and $\A(\varphi \R (\psi \U \gamma))
      \leftrightarrow \A (\varphi \R \alpha^{\U}_{\psi,\gamma})$,
%
    \item
      $\E((\varphi \R \psi) \U \gamma) \leftrightarrow
      \E(\alpha^{\R}_{\varphi,\psi} \U \gamma)$ and $\A ((\varphi \U \psi) \R
      \gamma) \leftrightarrow \A(\alpha^{\U}_{\varphi,\psi} \R \gamma)$.
    \end{compactenum}
\end{proposition}
Equivalences with $\F$ and $\G$ can be obtained from the equivalences with $\U$
and $\R$. Thus, even if it seems that \cCTL*[f] is able to syntactically specify
properties impossible to express with \polcCTL(p), the finite path
quantification indeed trivializes many of these properties.
Interestingly, both logics have an ability the other one lacks, but as we have already mentioned they share the same expressive power:  \polcCTL(p) enjoys past operators and thus the ability to reason backwards over the input tree, while \cCTL*[f] is able to nest $\U$ and $\R$ operators without a path quantifiers in between. The equivalence with \FO shows that these abilities yield the same expressiveness.
Indeed, notice that the \polcCTL(p) formula equivalent to the above cited \CTL*
formula $\E((p \lor q\U r)\U s)$, that makes \polcCTL(p) strictly more
expressive than \polcCTL~(see \Cref{polCTL stricly less than polCTLp} above),
preserves its meaning also in \cCTL*[f]. Basically, the syntactic restrictions
of \polcCTL(p) are semantically retained in \cCTL*[f]. 
Once again, we mention the completely additive fragment of the modal \MC to emphasize its relationship with these two logics and to better understand their behaviour. In \cite[Theorem 3.6]{CV14}, this fragment is shown equivalent to \PDL, while in \cite[Theorem 1]{Car15} it is shown that \PDL is equivalent to a fragment of Weak Monadic Chain Logic (\WMCL). Since \FO is itself a fragment of \WMCL, it makes sense that the temporal
logics equivalent to it enjoys such ``polarized'' properties, meaning that the existential and universal modalities are bounded to express dual properties, such as co-safety and safety ones.



\section{Hesitant and Weak Graded Tree Automata}
Here we study a general automaton model inspired by the class of 
\emph{Graded Tree Automata}, originally introduced in \cite[Section 3]{KSV02}. 
We progressively restrict the recognising power of this model, 
until we obtain two classes of automata provably 
equivalent to \polcCTL(p) and \cCTL*[f], and consequently to \FO. 

\medskip

\noindent{\bf Graded tree automata.}
A set of \emph{positive Boolean formula} over a set $X$ is denoted by $\mathcal{B}^+(X)$, and is composed of formulas of the form $\top, \bot, x, \theta \lor \theta, \theta \land \theta$, where $x \in X$.   

A \emph{Graded Alternating Parity Tree Automaton} (GTA) is a tuple 
$\mathcal{A} = \langle Q, \Sigma, \delta, q_I, \Omega\rangle$, 
where $Q$ is a finite set of states, $\Sigma$ is a finite alphabet, 
$q_I \in Q$ is an initial state, 
$\Omega: Q \rightarrow \mathbb{N}$ is a priority function, and 
$\delta: Q \times \Sigma \rightarrow \mathcal{B^+}(\{ \diamond_k, \square_k \mid k>0\} \times Q)$ 
is an alternating transition function.
We call \emph{atom} any pair $(\bigcirc_k, q)$, where $\bigcirc$ is either $\diamond$ or $\square$.
When $k=1$, we simply write $\diamond$ instead of $\diamond_k$ and $\square$ instead of $\square_k$. 

A \emph{run} of a GTA $\mathcal{A}$ on a $\Sigma$-labeled tree $\mathcal{T} = (T, \tau)$, 
is a $V_T \times Q$-labelled tree $\mathcal{R} = (R, \rho)$, where $R=(V_R,E_R)$ is a tree 
(possibly with some leaves) and 
$\rho: V_R \rightarrow V_T \times Q$ is a labeling function that associates every mnode of $R$ to a node of the input tree
and a state of the automaton. 
Before describing the additional conditions that must be satisfied by a run, we briefly explain how 
to evaluate a positive Boolean formula that may appear in the transition function of the automaton. 
We start by considering the case of atoms. 
Given a candidate run $\mathcal{R} = (R, \rho)$ on $\mathcal{T}=(T,\tau)$ and a node $s\in V_R$, 
with $\rho(s)=(n,q)$, we let:
\begin{compactitem}
    \item $s \models (\diamond_k, q')$ if $n$ has at least $k$ successors in $T$, say, $n_1, ..., n_k$, 
          and, for each $i=1,\dots,k$, $s$ has at least one successor in $R$ labeled by $(n_i,q')$;
    \item $s \models (\square_k, q')$ if for all but $k-1$ successors $n'$ of $n$ in $T$, 
          there is a successor of $s$ in $R$ that is labeled by $(n',q')$. 
\end{compactitem}
The satisfaction relation $\models$ is then extended to the constants $\top$ and $\bot$
and to the Boolean connectives $\land$ and $\lor$ in the obvious way,
e.g.~by letting $s\models\top$ when $s$ is a leaf.
Now, the additional conditions that must be satisfied by a run $\mathcal{R} = (R,\rho)$ 
of $\mathcal{A}$ on $\mathcal{T} = (T, \tau)$ are as follows:
\begin{inparaenum}[(a)]
    \item $\rho(\epsilon_R)=(\epsilon_T,q_I)$,
    \item for every node $s\in R$, with $\rho(s) = (n,q)$, $s \models \delta(q,\tau(n))$.
\end{inparaenum}
Given such a run $\mathcal{R} = (R, \rho)$ and a maximal path $\pi$ in it (which can be finite or infinite), 
we denote by $\inf(\pi)$ the set of states visited infinitely often along $\pi$,
and we say that $\pi$ is \emph{accepting} if either it terminates in a leaf (implying that
there is a transition reaching $\top$), or it is infinite and the highest priority 
assigned by $\Omega$ to the states in $\inf(\pi)$ is even.
Finally, a run $\mathcal{R}$ is \emph{successful} if every maximal path in it is accepting. 

\medskip \noindent{\bf Weak and Hesitant tree automata.}
We are interested in a proper subclass of GTAs, characterized by a \emph{hesitant} 
partition on the state set and a \emph{weak} acceptance condition. 
We define the weak acceptance condition first, following \cite{MSS86}. 

A \emph{Weak Graded Alternating Parity Tree Automaton} $\mathcal{A} = \langle Q, \Sigma, \delta, q_I, \Omega\rangle$ (WGT) is a GTA such that there is a partition of $Q$ into disjoint non-empty sets (from now on, \emph{components}) \{$Q_1, ...., Q_k$\} and a partial order $\leq$ such that the transitions from a state in $Q_i$ can only lead to states in $Q_i$ or to states in a component with lower order. The acceptance condition is weak because we enforce that every component either contains only states marked even by the priority function or contains only states marked odd. 

A \emph{Hesitant Weak Graded Alternating Tree Automaton} (HWGT) $\mathcal{A} = \langle Q, \Sigma, \delta, q_I, \Omega\rangle$ 
is a WGT such that each component is of one of the following three types:
\vspace{-0.5ex}
\begin{compactitem}
    \item $Q_i$ is \emph{existential}, if for all $\sigma \in \Sigma$ and for all $q, q' \in Q_i$, 
    $q'$ can appear in the disjunctive normal form of $\delta(q,\sigma)$ only in an atom $(\diamond, q')$, 
    and only disjunctively related to other atoms with states in $Q_i$;
    \item $Q_i$ is \emph{universal}, if for all $\sigma \in \Sigma$ and for all $q, q' \in Q_i$, 
    $q'$ can appear in the conjunctive normal form of $\delta(q, \sigma)$ only in an atom $(\square, q')$,
    and only conjunctively related to other atoms with states in $Q_i$;
    \item $Q_i$ is \emph{transient}, if for all $\sigma \in \Sigma$ and for all $q,q' \in Q_i$, 
          $q'$ does not appear in any atom in $\delta(q,\sigma)$.
\end{compactitem}
\vspace{-0.5ex}
Note that in the existential and universal components $\diamond_k$ and $\square_k$ must have $k = 1$, if they are paired with a state of the component. 
    This is the \emph{hesitant} constraint on the partition of the state set, introduced in \cite[Section 5.1]{KVW00}. Clearly, every path of a given run will eventually get stuck in an existential or a universal component. We let every state in a universal (resp., existential) component be marked even (resp., odd) by the priority function. Let $Q_i$ be the component in which a path $\pi$ of a given run $\mathcal{R}$ gets stuck: if $Q_i$ is universal, $\pi$ satisfies the accepting condition if it visits infinitely often a state in $Q_i$; if $Q_i$ is existential, $\pi$ satisfies the accepting condition if it visits finitely often every state in $Q_i$. Notice that this accepting condition is weak, since it basically states that the universal components are entirely accepting and the other components are entirely rejecting.
It is possible to have a HWGT $\mathcal A$ with a positive boolean formula $\theta$ as initial condition instead of a state. It is then possible to convert it to an automaton with an initial state by having as initial an extra state $q_\theta$, such that $\delta(q_\theta, \sigma) = \theta$ for every $\sigma \in \Sigma$. 
Such an automaton will be denoted by $\mathcal{A}^\theta$. Given a HWGT $\mathcal{A}$ and a state $q \in Q$, we denote by $\mathcal{A}^q$ the automaton $\mathcal{A}$ where $q$ is the starting state.

Summing up, we will work with a HWGT $\mathcal{A} = \langle Q, \Sigma, \delta, q_I, \Omega \rangle$, with an hesitant partition on the state set and a weak acceptance condition, stating that every state in the existential components are only visited finitely often or some state in the universal components is visited infinitely often. However, this class of automata is still not \emph{weak} enough to be equivalent to \FO. Indeed, note the striking similarity of HWGTs with the \emph{Additive Weak Parity Automata} (\emph{AWA}) introduced in \cite[Section IV]{Car15} and proven equivalent to \WMCL on trees. This suggests that HWGTs as above defined are equivalent to this logic, too. In the following, we will further restrict their recognising power and yield equivalence results with the two temporal logics above introduced.



\section{Automaton-based Characterization of \texorpdfstring{\polcCTLp}{polarized cCTLp}}
What currently prevents us from saying that HWGTs and \polcCTLp are equivalent are the following two observations: an HWGT cannot reason upward along the input tree and its components have no restriction other than being purely accepting or purely rejecting. In this section we will solve both these problems in a straightforward way, and prove that the obtained class of automata is equivalent to \polcCTLp.

\medskip

\noindent{\bf Two-Way Linear HGTA.}
A \emph{Two-Way HWGT} (\emph{2HWGT}) 
is defined exactly like a HWGT, namely, as a tuple
$\mathcal{A} = \langle Q, \Sigma, \delta, q_I, \Omega \rangle$,
but where 
$\delta: Q \times (\Sigma \times \{ 0,1 \}) \rightarrow  \mathcal{B}^+(\{\diamond_k, \square_k, -1\} \times Q)$.
The idea is that the automaton can send states
towards the successors of the currently visited node, as well as towards its parent, if it exists. 
These transitions are determined, as usual, from the current state and from the label of the visited node, 
but they might also depend on whether the head is at the root or not --- for this, we explicitly mark the 
nodes of the input tree with a flag that distinguishes the root from the other nodes.

%
%
%
%
Given a run $\mathcal{R} = (R,\rho)$ of a 2HWGT $\mathcal{A}$ over a $\Sigma$-labelled tree as above, and a node $s$ in $V_R$ such that $\rho(s) = (n',q')$, the semantics of an atom $(-1,q) \in \mathcal{B}^+(\{\diamond_k, \square_k, -1\} \times Q)$, 
for some $q \in Q$, is as follows:


\begin{compactitem}
\item $s \models (-1, q)$ \iff $n'$ has a parent $n''$ and $s$ has a parent labelled by ($n'',q$) -- note in particular that when the head is at the root this atom is always false.
\end{compactitem}
Moreover, we enlarge the hesitant types adding to the above defined transient, existential and universal the following. Given a 2HWGT $\mathcal{A}$ and a component $Q_i$ of $\mathcal{A}$:
\begin{compactitem}
    \item $Q_i$ is \emph{upward}, if for all $\sigma \in \Sigma$ and for all $q,q' \in Q_i, q'$ can appear in $\delta(q,\sigma)$ only in the form $(-1,q)$.
\end{compactitem}

Finally, a \emph{linear} HWGT, or HLGT for short, is a HWGT
$\mathcal{A} = \langle Q, \Sigma, \delta, q_I, \Omega \rangle$ 
in which the partitioning of the state set is composed entirely of \emph{singletons}. 
This latter condition is reminiscent of a restriction on Additive Weak Parity Automata (AWA) (already cited above at the end of section 4)
introduced in \cite[Definition 3.3.1]{For19}, and giving rise to the so-called
\emph{Antisymmetric} AWA. In that paper, the author states the expressive equivalence 
between AWA and \WMCL, and conjectures an equivalence between Antisymmetric AWA and \FO
(he was only able to prove that \FO is at least as expressive as Antisymmetric AWA).
This conjectured correspondence will be discussed again at the end of the section.


Finally, merging the two restrictions together, we get the class of 2HLGT, 
i.e., two-way automata in which every state is either transient, existential, 
universal or upward. We will now prove that 2HLGT are equivalent to \polcCTLp.

\subsection{Equivalence of 2HLGT and \texorpdfstring{\polcCTLp}{polarized cCTLp}}
Here we prove that given a 2HLGT, it is always possible to obtain a \polcCTLp
formula equivalent to it. The approach is a combination of the translation
procedure of \cite[Theorem 6]{LT00} and \cite[Theorem 3.1]{BS18}.
\begin{theorem}
    Given a 2HLGT $\mathcal{A}$ over $\Sigma = 2^{AP}$, for a set of atomic propositions $AP$, it is always possible to obtain a \polcCTLp formula $\varphi_\mathcal{A}$ such that $\mathcal{L}(\mathcal{A}) = \mathcal{L}(\varphi_\mathcal{A})$.
\end{theorem}

\begin{proof}
    Fix a 2HLGT $\mathcal{A} = \langle Q, \Sigma, \delta, q_I, \Omega \rangle$. We will show how to translate any given $q \in Q$ by induction over the order of the singletons in the partition. This will yield the result, by the translation of $q_I$. For a given state $q$ or a positive boolean formula $\theta$, we will use in the following $\chi(q)$ and $\chi(\theta)$ for the \polcCTLp formulas equivalent to $\mathcal{A}^q$ and $\mathcal{A}^\theta$, respectively.

    First of all, we translate any $\sigma \in \Sigma$ by the \polcCTLp formula $\psi_\sigma = \bigwedge\nolimits_{\text{$\apElm \in \sigma$}} \apElm \land \bigwedge\nolimits_{\text{$\apElm \notin \sigma$}} \neg\apElm$.
By construction, for every $q \in Q$ and every $\sigma \in \Sigma$, if $q' \in
\delta(q,\sigma)$ and $q' \neq q$, we know that $q'$ is of lower order than
$q$ with respect to the partition of the state set. By inductive hypothesis, we
could assume to have a \polcCTLp formula $\chi(q')$ for each one of
them, if it wasn't for atoms. Indeed, since \polcCTLp is closed under
$\lor$ and $\land$, and $\top$ and $\bot$ are already allowed formulas of the
logic, we can assume to translate this kind of formulas without any problem.
With regards to atoms, however, the translation is also easy:
$(\diamond_k, q)$ is translated as $\D^k\chi(q)$, ($\square_k, q)$ is translated
as $\neg \D^k\neg \chi(q)$ and ($-1,q) = \Y \chi(q)$. We just have to show how to
obtain $\chi(q)$. This is done in the following way. Having fixed a $q \in Q$
and a $\sigma \in \Sigma$, $\delta(q,\sigma)$ is a positive boolean formula
$\theta$. Moreover, if we ignore a (possible) atom in wich $q$ itself occurs, it
is a positive boolean formula composed entirely of states of order lower than
$q$, denoted by $\chi(\theta_{q, \sigma})$. The subscript says that it is the
formula obtained by $\delta(q,\sigma)$. Now, we define the recursion to obtain $\chi$. This is done by induction on the order on the components of the
partition, that is well-founded.

Given $q \in Q$ and $\sigma \in \Sigma$, we know that $q$ can only be of a specific type because the automaton is hesitant. For any one of these types we can present $\delta(q,\sigma)$ in a precise way, as follows:
\begin{compactitem}
    \item if $q$ is transient, $\delta(q, \sigma$) is of the form $\theta'$,
    \item if $q$ is existential, $\delta(q, \sigma$) is of the form (($\diamond, q) \land \theta) \lor \theta')$,
    \item if $q$ is universal, $\delta(q, \sigma$) is of the form (($\square, q) \land \theta) \lor \theta')$,
    \item if $q$ is upward, $\delta(q, \sigma$) is of the form (($-1, q) \land \theta) \lor \theta')$.
\end{compactitem}
We impose that $q$ does not appear neither in $\theta$ nor $\theta'$ and this is always possible because if there are more instances of $q$ in $\delta(q,\sigma)$ we can always rewrite it with just one. The reasoning behind this form is that $\theta$ is the positive boolean formula that appears conjunctively related with $q$ itself, while $\theta'$ is only disjunctively related to $q$. This means that whenever $\theta$ is true, the automaton has to keep staying in state $q$, while whenever $\theta'$ is true, it must leave $q$.

Recall that by induction we can assume to have $\chi(\theta_{q,\sigma})$ and $\chi(\theta'_{q,\sigma})$. Thanks to this, we can obtain two formulas related to $q$, that tells if the automaton stays in the state or leaves it. Namely:

{

  \centering

    $\gamma_q = \bigvee\nolimits_{\text{$\sigma \in \Sigma$}} \psi_\sigma \land \chi(\theta_{q, \sigma})$
    \quad \quad
    $\gamma'_q = \bigvee\nolimits_{\text{$\sigma \in \Sigma$}} \psi_\sigma \land \chi(\theta'_{q, \sigma})$

}

What is missing is how to translate a single state:
\begin{compactitem}
    \item if $q$ is transient, $\varphi_q = \gamma'_q$,
    \item if $q$ is existential, $\varphi_q = \E(\gamma_q \U \gamma'_q)$,
    \item if $q$ is universal,
    $\varphi_q = \A(\gamma_q \W \gamma'_q)$,
    \item if $q$ is upward, $\varphi_q = \gamma_q\S\gamma'_q$.
\end{compactitem}
\end{proof}
\vspace{-4.5ex}
\begin{proposition} \label{correctness of translation from 2HLGT to logic}
    For any $q \in Q$, $\mathcal{L(\mathcal{A}}^q) = \mathcal{L}(\chi(q))$.
\end{proposition}

The other direction of the translation is classic and can be obtained by adapting in a straightforward way the proof of \cite [ Section 14.7.2]{DGL16}.

\begin{theorem} \label{translation of logic to 2HLGT}
    Given a \polcCTLp formula $\varphi$ one can construct a 2HLGT $\mathcal{A}_\varphi$ such that $\mathcal{L}(\varphi) = \mathcal{L}(\mathcal{A}_\varphi)$.
\end{theorem}

Clearly, the equivalence between 2HLGT and \polcCTLp is allowed by backward transitions. If one removes the possibility of going ``up" the input tree, one also loses the expressive power given by the past operators of the logic. By inspecting the translations, it is clear that the two-way head movements and the past temporal operators are intertwined, and that if one removes them both from the two formalisms there is still a possibility of translations between the formalisms thus obtained, yielding the following lemma.
\begin{lemma}
    HLGT and \polcCTL are equivalent formalisms.
\end{lemma}

By combining this with Propostion \ref{polCTL stricly less than polCTLp}, one obtains that HLGT are strictly less expressive than \FO. The following conjecture implies that the class of Antisymmetric AWA introduced in \cite{For19} and above mentioned is strictly less expressive than \FO over trees.

\begin{conjecture}
HLGT and Antisymmetric Additive-Weak Parity Automata are equivalent
formalisms.
\end{conjecture}



\section{Automaton-based Characterization of \texorpdfstring{\cCTL*[f]}{cCTLf*}}
In this section, we complete the picture of logics-vs-automata correspondences by
showing a class of automata, in fact, restrictions of HWGTs, that are expressively 
equivalent to \cCTL*[f]. 
We will follow the approach from \cite{BBMP24}, with minor cosmetic changes.
We shall mostly focus on the translation from the restriction of HWGTs 
(to be defined soon) to equivalent \cCTL*[f] formulas, which is the interesting
part of the proof.

\medskip

\noindent{\bf Linearization of components.}
The first step consists in focusing on the non-transient components 
of an HWGT, and thinking of them as suitable word automata over an expanded alphabet.
We remark that in this step, we move from automata that process trees to automata
that process words. This abstraction is possible thanks to the restrictions
imposed to an HWGT.

\begin{definition} \label{Linearization}
Let $\cA = \langle Q, \Sigma, \delta, q_I, (Q_\forall,Q_\exists) \rangle$ 
be an HWGT and let $B$ be the set of all atoms that occur in the images 
of its transition function.
For each component $Q_i$ of $\cA$ and each state $q\in Q_i$, let $q_{\mathrm{exit}}$
be a fresh state and let $\cA_{Q_i, q} = \langle Q_i\uplus\{q_{\mathrm{exit}}\}, \Sigma \times 2^B, \delta', q, Q_i \rangle$
be the word automaton obtained from $\cA$ by restricting the state set to $Q_i$ and adding $q_{\mathrm{exit}}$, 
declaring $q$ to be the new initial state, annotating the input letters with subsets of $B$, 
and redefining the transition relation and the acceptance condition as follows:
\begin{compactenum}
    \item if $Q_i$ is an existential component, then $\cA_{Q_i, q}$ is a 
          non-deterministic coBüchi automaton (NCA),
          with $Q_i$ as set of rejecting states 
          (implying that the only successful runs are those that eventually exit $Q_i$), and with transitions defined by 
          \begin{compactitem}
            \item $q_{\mathrm{exit}} \in \delta'(q',(\sigma,C))$ (for $q'\in Q_i\uplus\{q_{\mathrm{exit}}\}$, $\sigma\in\Sigma$, $C\subseteq B$) 
                  whenever $q'=q_{\mathrm{exit}}$ or 
                  the disjunctive normal form of $\delta(q',\sigma)$ contains a clause of the form $\bigwedge C$
                  (note that, thanks to the restrictions imposed to an HWGT,
                  the states appearing in $C$ have all lower order than those in $Q_i$);
            \item $q'' \in \delta_{Q_i}(q', (\sigma,C))$ (for $q',q''\in Q_i$, $\sigma\in\Sigma$, $C\subseteq B$) 
                  whenever the disjunctive normal form of $\delta(q',\sigma)$ contains
                  a clause of the form $(\diamond, q'') \land \bigwedge C$;
          \end{compactitem}
    \item dually, if $Q_i$ is a universal component, then $\cA_{Q_i, q}$ is a universal Büchi automaton (UBA),
          with $Q_i$ as set of accepting states 
          (implying that the successful runs are those that remain inside $Q_i$), and with transitions defined by 
          \begin{compactitem}
            \item $q_{\mathrm{exit}} \in \delta'(q',(\sigma,C))$ (for $q'\in Q_i\uplus\{q_{\mathrm{exit}}\}$, $\sigma\in\Sigma$, $C\subseteq B$) 
                  whenever $q'=q_{\mathrm{exit}}$ or 
                  the conjunctive normal form of $\delta(q',\sigma)$ contains a clause of the form $\bigvee C$;
            \item $q'' \in \delta_{Q_i}(q', (\sigma,C))$ (for $q',q''\in Q_i$, $\sigma\in\Sigma$, $C\subseteq B$) 
                  whenever the conjunctive normal form of $\delta(q',\sigma)$ contains
                  a clause of the form $(\square, q'') \vee \bigvee C$.
          \end{compactitem}
\end{compactenum}
We denote by $B_{Q_i}$ the set of subsets $C$ of $B$ for which there is $q'\in Q_i$
with $\delta'(q', (\sigma,C)) \cap Q_i \neq \emptyset$ -- intuitively, these are the annotations
that allow the automaton $\cA_{Q_i,q}$ to enter a lower component.
\end{definition}

The word languages recognized by the $\cA_{Q_i,q}$'s are quite simple.
Specifically, when $Q_i$ is an existential component, $\cA_{Q_i,q}$ accepts by exiting $Q_i$, 
implying that it recognizes a \emph{co-safety language},
namely, a language of the form $F\Sigma^\omega$, for some finite $F\subseteq\Sigma^*$.
Dually, when $Q_i$ is universal, $\cA_{Q_i,q}$ accepts by staying inside $Q_i$,
implying that it recognizes a \emph{safety language}, i.e.~the complement of a
co-safety language.
Summing up, with the above definition we have achieved two crucial goals.
First, we moved from tree automata to word automata, thus enabling the use
of characterizations of subclasses of languages.
Second, we encoded the behaviour of each component as a (co-)safety language,
which is essentially a property referred to \emph{finite} prefixes. 
This justifies the existence of a translation from (a suitable restriction of) HWGTs
to \cCTL*[f], the variant of $\cCTL*$ with quantification on finite paths.

\medskip

\noindent{\bf Counter-free word automata.}
The next step is to bound the expressive power of these word automata to get \FO expressiveness, that over words is equivalent to \LTL, as already remarked. The goal is then to translate each automaton $\cA_{Q_i,q}$ to an 
equivalent \LTL formula and, combining together the formulas thus obtained, we will get 
a single \cCTL*[f] formula that defines the entire language of trees 
recognized by the original HWGT.

Of course not all (co-)safety languages are \LTL-definable. 
This means that, in order to enable the desired translation, 
we need to first identify which HWGTs produce component automata 
$\cA_{Q_i,q}$ that are within the expressiveness of \LTL. 
For this, we shall rely on the well-known characterization of \LTL 
in terms of languages recognized by counter-free automata \cite[Theorem 1.1]{DG08a}.
In view of this, it is tempting to simply require that all non-transient
components of the given HWGT induce counter-free word automata $\cA_{Q_i,q}$:
this would indeed guarantee that each $\cA_{Q_i,q}$ translates to an
equivalent \LTL formula.
However, a straightforward adaptation of \cite[Example 5.1]{BBMP24} 
shows an HWGT in which all non-transient components induce
counter-free word automata, and yet the language recognized by the HWGT
is not definable in \FO.
In particular, this means that the counter-free condition on the components 
will not guarantee the possibility of combining together the \LTL 
formulas obtained in the previous step so as to obtain a single \cCTL*[f] 
equivalent to the HWGT.
However, in \cite{BBMP24} it is also shown that one remains in \FO if,
in addition to the counter-free condition on the components,
a mutual exclusion property is also enforced:


\begin{definition} \label{mutual exclusion}
An HWGT $\mathcal{A}$ satisfies the \emph{mutual exclusion} property 
if for every non-transient component $Q_i$ and every $C \neq C' \in B_{Q_i}$, 
there are atoms $\alpha\in C$ and $\beta \in C'$ such that 
$\mathcal{L}(\mathcal{A}^\alpha)$ is the complement of 
$\mathcal{L}(\mathcal{A}^\beta)$.
\end{definition}

We denote by HWGT$_{cf}$ the subclass of HWGTs that satisfy the mutual exclusion
property and the counter-free condition on the word automata induced by non-transient components.
Accordingly, from now on, we assume that $\cA$ is an HWGT$_{cf}$.

Let us now inspect the translation from the word automata 
$\cA_{Q_i,q}$'s to equivalent \LTL formulas. 
Since the $\cA_{Q_i,q}$'s recognize safety or co-safety languages, 
it is useful to first recall the \emph{safe} and \emph{co-safe fragments} of \LTL,
that capture precisely these languages within \LTL. 
The safe (resp., co-safe) fragment of $\LTL$, denoted \SafeLTL~(resp., \CosafeLTL), 
allows the usual atomic propositions and Boolean connectives (except for negation), the next operator $\X$, 
and the release operator $\R$ (resp., the until operator $\U$).
%
%
On the automaton side, we also recall that \SafeLTL~(resp., \CosafeLTL) languages are recognized
precisely by the so-called \emph{looping}, counter-free deterministic Büchi (resp., deterministic coBüchi) automata
\cite{BLS22}, where looping means that all states are accepting (resp., rejecting), with the  exception of a single sink state. Building upon this result, the following can be obtained with an adaptation of the argument in \cite[Theorem 5.1]{MH84}:  
%
%
%
\begin{lemma} \label{equivalence of automata with safetyltl} \sloppypar
\SafeLTL~(resp., \CosafeLTL) and counter-free, looping, universal Büchi (resp., non-deterministic coBüchi) 
automata are expressively equivalent formalisms.
\end{lemma}



\subsection{Equivalence of HWGT\texorpdfstring{$_{cf}$}{cf} and \texorpdfstring{\cCTL*[f]}{cCTLf*}}

In this section, we prove the equivalence of HWGT$_{cf}$ and the logic \cCTL*[f]. 
This equivalence 
yields a normal form for \cCTL*[f] formulas, which turns out to be 
a sort of \emph{polarized} \cCTL*.

\begin{theorem} \label{translation of HWGT in logic}
    Let $\mathcal{A}$ be a HWGT$_{cf}$ over the alphabet $\Sigma = 2^{AP}$. Then, one can construct a \cCTL*[f] formula $\varphi_\mathcal{A}$ such that $\mathcal{L}(\mathcal{A}) = \mathcal{L}(\varphi_\mathcal{A})$.
\end{theorem}
\begin{proof}
   Let $\mathcal{A} = \langle Q, \Sigma, \delta, q_I, (Q_\forall,Q_\exists) \rangle$ be a HWGT$_{cf}$. For every state $q \in Q$, we can obtain a \cCTL*[f] formula $\chi(q)$ such that $\mathcal{L}(\mathcal{A}^q) = \mathcal{L}(\chi(q))$ and the theorem follows by setting $q$ to be $q_I$. The proof is by induction on the order of the components. Moreover, for every $\sigma \in \Sigma$, we denote by $\psi_\sigma$ the \cCTL*[f] formula: $\bigwedge\nolimits_{\text{$\apElm \in \sigma$}} \apElm \land \bigwedge\nolimits_{\text{$\apElm \notin \sigma$}} \neg\apElm$.
So, let's fix a component $Q_i$, and let's say $q \in Q_i$. We consider the case in which $Q_i$ is universal: the transient case is straightforward and the existential case can be restored by a dualization argument. Thus, suppose $Q_i$ is universal. Then, we can focus on the counter-free looping UBA $\mathcal{A}_{Q_i, q} = \langle Q_i \uplus \{\mathrm{q_{exit}}\}, 2^{AP} \times B_{Q_i},\delta_{Q_i}, q, Q_i\rangle$, as given in Definition \ref{Linearization}.  Moreover, since every $C \in B_{Q_i}$ is composed of atoms whose states belong to components of lower order than $i$, by induction hypothesis we can assume to have a \cCTL*[f] formula $\chi(C)$ for every  $C \in \Gamma_{Q_i}$, such that 
$\mathcal{L}(\mathcal{A}^{\bigvee C}) = \mathcal{L}(\chi(C ))$. 
Since $\mathcal{A}$ satisfies the mutual exclusion property, it also holds that for every distinct 
$C, C' \in \Gamma_{Q_i}$, $\mathcal{L}(\chi(C)) \cup \mathcal{L}(\chi(C')) = \mathcal{T}_\Sigma$.

       For every $C \in B_{Q_i}$, we define a fresh atomic proposition $p_C$. We denote by $AP_B$ the union of $AP$ with all these new atomic propositions. Now, we can see the Büchi automaton $\mathcal{A}_{Q_i, q}$ as $\mathcal{A}_B$ = $\langle Q_i \uplus \{\mathrm{q_{exit}}\}, 2^{AP_B},\delta_B, q, Q_i\rangle$. Basically, we just hide $B_{Q_i}$ in the alphabet. The only thing that is modified from $\mathcal{A}_{Q_i, q}$ to $\mathcal{A}_B$ is the transition function, i.e., $\delta_B(q, \sigma_C) =\delta_{Q_i}(q,(\sigma,C))$ if $\sigma_C = \sigma \cup p_C$ (where $\sigma \subseteq AP$), otherwise $\delta_B(q, \sigma_C) = \emptyset$. Clearly, these little changes do not modify any property of the starting automaton, so $\mathcal{A}_B$ is still counter-free and looping. By Lemma \ref{equivalence of automata with safetyltl}, we can assume to have a \SafeLTL formula $\psi_B$, such that $\mathcal{L}(\mathcal{A}_B) = \mathcal{L}(\psi_B)$. Now, we have to find a way to link this formula to the tree language defined by the automaton. First, we note that the following holds by construction.

       \emph{Fact.} Let $\mathcal{A}$ be an HWGT$_{cf}, Q_i$ be a universal component of $\mathcal{A}$ and $q \in Q_i$. Then, for each input tree $\mathcal{T} = (T, \tau), \mathcal{T} \in \mathcal{L}(\mathcal{A}^q)$ \iff for every infinite path $\pi$ of $\mathcal{T}$ starting at the root, there is a word $\xi$ of the form $\xi = (\tau(\pi(0)), C_0), (\tau(\pi(1)), C_1), (\tau(\pi(2)), C_2),...$ , such that either $\xi \in \mathcal{L}(\mathcal{A}_{Q_i, q})$ or there is a a position $i = (\tau(\pi(i)), C_i)$ such that  
       $\mathcal{T}_{\pi(i+1)} \in \mathcal{L}(\mathcal{A}^{\bigvee C_{i+1}})$, 
       where $\mathcal{T}_{\pi(i+1)}$ is the subtree of $\mathcal{T}$ rooted at node $\pi(i+1)$.

        Note this about the above claim: it basically states that a node such that its subtree is in 
        $\mathcal{L}(\mathcal{A}^{\bigvee C_i})$ 
        \emph{releases} the path on which it is from the satisfaction of the constraint imposed by 
        $\mathcal{A}_{Q_i, q}$. Moreover, we saw above that by induction hypothesis we have for every 
        $C \in B_{Q_i}$ a \cCTL*[f] formula $\chi(C)$ such that 
        $\mathcal{L}(\mathcal{A}^{\bigvee C}) = \mathcal{L}(\chi(C))$. 
        Combining this with the fact above, we get the following claim about the tree language defined by $\mathcal{A}^q$. This characterization can be used to prove the correctness of the translation.

       \emph{Fact.} For each $\Sigma$-labelled tree $\mathcal{T}$, $\mathcal{T} \in \mathcal{L}(\mathcal{A}^q)$ \iff for every infinite path $\pi$, there is an infinite word $\xi$ over 2$^{AP_B}$ such that either $\xi \models \psi_B$, or there is $i \geq 0$ such that $\xi(i) \cap AP = \pi(i)$, for all $p_C \in \xi(i), \mathcal{T}, \pi, i \models \chi(C)$ and there is a unique $C \in B_{Q_i}$, such that $p_C \in \xi(i)$.

       Now, replace every $p_C$ and $\neg p_C$ in $\psi_\Gamma$ by the \cCTL*[f] formulas $\chi(C)$ and $\bigvee\nolimits_{C' \in \Gamma_{Q_i} \setminus \{C\}}\chi(C')$, respectively. 
       The formula thus constructed is denoted $f(\psi_B$). 
       We can finally define the formula $\chi(q)$:

       {

         \centering

       $\chi(q) = \A\left(f(\psi_B) ~\W~ \bigvee\nolimits_{C \in B_{Q_i}} \chi(C)\right)$

   }

       The formula basically states that on every path the satisfaction of the formula defined by the word automaton can only be released by the satisfaction of a 
       $\bigvee C$.
\end{proof}
\begin{proposition} \label{correctness translation of HWGT to logic}
    For any $q \in Q_i$, where $Q_i$ is a universal component, $\mathcal{L}(\mathcal{A}^q) = \mathcal{L}(\chi(q))$.
\end{proposition}

The translation from a \cCTL*[f] formula to a HWGT$_{cf}$ is basically the same provided in \cite[Theorem 5.9]{BBMP24}. 

\begin{theorem} \label{translation from logic to HWGT}
Given a \cCTL*[f] formula $\varphi$, one can construct a HWGT$_{cf}$ $\mathcal{A}_\varphi$ 
such that $\mathcal{L}(\varphi) = \mathcal{L}(\mathcal{A}_\varphi)$.
\end{theorem}



\section{Normal Forms of Temporal Logics}
In the previous sections, we have proved the equivalence of the two logics considered in Section 3, i.e., \polcCTLp and \cCTL*[f], with two classes of automata. In particular, thanks to Proposition \ref{ctlf equivalences} and the class of automaton proven equivalent to \cCTL*[f], it was possible to highlight the \emph{semantic} behaviour of the latter. Namely, whenever an existential path quantification is involved, a \cCTL*[f] formula can only express a co-safety property, while, dually, whenever a universal path quantification is involved, it can only express a safety property. These observations give rise to the following normal form, that captures \emph{syntactically} the \emph{semantic} content provided by the finite path quantification.

\begin{lemma}
    For any \cCTL*[f] formula, there is an equivalent formula generated by the grammar
    \vspace{-1.5ex}
    \setlength{\jot}{0.25pt}
\begin{align*}
    \varphi \;&::=\; p
    \;\mid\; \neg \varphi
    \;\mid\; \varphi \lor \varphi
    \;\mid\; {\D^n\varphi}
    \;\mid\; {\E\psi}
  \\
  \psi \;&::=\; \varphi
    \;\mid\; \psi \lor \psi
    \;\mid\; \psi \land \psi
    \;\mid\; {\X \psi}
    \;\mid\; {\psi \U \psi}
\end{align*}

\end{lemma}

Note that this grammar allows to state that \E \xspace is only followed by \CosafeLTL and, by the use of negation, that \A \xspace is only followed by \SafeLTL. Moreover, the difference between finite and infinite path quantification becomes redundant. Indeed, every \emph{finite} path property can also be seen as an \emph{infinite} path property and vice versa. This implies that the normal form of \cCTL*[f] is nothing else than a \emph{polarized} version of \cCTL*, that we will denote by \polcCTL*, creating a symmetry with Schlingloff's work and also showing that the semantic content provided by finite path quantification is useless when one restricts the syntax as above. To conclude, the following is well known.
\begin{proposition} \cite{CMP92}
    \SafeLTL~(resp., \CosafeLTL) and \LTL(p) formulas of the form $\G\varphi$ (resp., $\F\varphi)$, where $\varphi$ is a formula using only past temporal operators, are equivalent formalisms.
\end{proposition}
This suggests a normal form also for \polcCTL(p). Since \polcCTL(p) and \CTL[f] are equivalent formalisms, 
\polcCTL(p) can express co-safety properties existentially and safety properties universally. Combining this with the proposition above, we get the following normal form for \polcCTL(p).

\begin{lemma}
    For any \polcCTL(p) formula, there is an equivalent formula generated by the grammar
    \vspace{-1.5ex}
    \setlength{\jot}{0.25pt}
\begin{align*}
    \varphi \;&::=\; p
    \;\mid\; \neg \varphi
    \;\mid\; \varphi \land \varphi
    \;\mid\; {\D^n\varphi}
    \;\mid\; {\E\F\psi}
  \\
  \psi \;&::=\; \varphi
    \;\mid\; \psi \lor \psi
    \;\mid\; \psi \land \psi
    \;\mid\; {\Y \psi}
    \;\mid\; {\psi \S \psi}
\end{align*}
\end{lemma}



\section{Conclusions}

In this work, we provided two automaton-based characterisations of the temporal
logics \polcCTL(p) and \cCTL*[f], both of which are known to be equivalent to \FO over
infinite trees.
These gives us two corresponding characterisations of \FO.
The automata-theoretic perspective reveals two distinctive features of \FO in
this setting:
\begin{inparaenum}[(a)]
\item
  when expressing existential properties over paths, it can capture only
  co-safety properties of the node sequences along those paths, whereas
  universal path quantification allows it to express only safety properties;
%
%
\item
  every formula can be normalised into a Boolean combination of formulae where
  only the variable bound to the outermost quantifier is independent, while all
  others depend on the variable(s) quantified first.
\end{inparaenum}
These insights were obtained by establishing corresponding normal forms for
\emph{\polcCTL(p)} and \emph{\cCTL*[f]}, each derived from its associated class of
automata.

Despite these advancements, several open problems remain.
First, while~\cite{Sch92a} shows that \polcCTLp is equivalent to \FO,
and~\cite{MR03} establishes that \cCTL* corresponds to \MPL, no analogous
result is known for \cCTL(p).
As shown in \Cref{CTL+p more expressive than polctlp}, \cCTL(p) is strictly more
expressive than \polcCTLp, while it is well known that it is strictly less
expressive than \cCTL*.
Determining the exact expressive power of \cCTL(p) remains an important open
question.
Second, although the two classes of automata we introduced are equivalent,
current translations between them passes through \FO, which leads to a
non-elementary blowup.
An interesting direction for future work would be to develop direct translations
between the two automata classes, bypassing the intermediate \FO encoding.
Third, we observe a strong similarity between the normal form of \polcCTLp and
the logic studied in~\cite{Boj09}.
While \polcCTLp supports both counting and the \S \xspace operator, the logic
in~\cite{Boj09} includes only the past-time version of \F.
Importantly, the definability problem for that logic is decidable over finite
trees.
Investigating whether the same holds for \polcCTLp in the finite-tree setting
would be a valuable contribution.
Finally, a more immediate objective is to prove the conjectured equivalence
between HLGTs and Antisymmetric Additive-Weak Parity Automata.

\bibliographystyle{eptcs}
\bibliography{References}
\end{document}